\documentclass[11pt]{article}
\usepackage{amsmath,amssymb,amsfonts,amsthm,latexsym,epsfig,hyperref,array,enumerate,geometry,graphicx,url,tikz,float,lipsum}
\usepackage[all]{xy}
\newtheorem{theorem}{\bf Theorem}
\newtheorem{lemma}[theorem]{\bf Lemma}
\newtheorem{proposition}[theorem]{\bf Proposition}
\newtheorem{corollary}[theorem]{\bf Corollary}
\newtheorem{remark}[theorem]{\bf Remark}
\newtheorem{definition}[theorem]{\bf Definition}

\usepackage[labelfont=bf]{caption}

\usepackage [english]{babel}
\usepackage [autostyle, english = american]{csquotes}
\MakeOuterQuote{"}

\begin{document}
\title{\bf On APN functions and their derivatives}

\author{Augustine Musukwa}
% First names are abbreviated in the running head.
% If there are more than two authors, 'et al.' is used.
\date{}
\maketitle  

\begin{center}{Mzuzu University, P/Bag 201, Mzuzu 2, Malawi \\augustinemusukwa@gmail.com } \end{center}            % typeset the header of the contribution
\begin{abstract}
We determine a connection between the weight of a Boolean function and the total weight of its first-order derivatives. The relationship established is used to study some cryptographic properties of Boolean functions. We establish a characterization of APN permutations in terms of the weight of the first-order derivatives of their components. We also characterize APN functions by the total weight of the second-order derivatives of their components. The total weight of the first-order and second-order derivatives for functions such as permutations, bent, partially-bent, quadratic, plateaued and balanced functions is determined.
\end{abstract}


\section{Introduction}

The nonlinearity and differential uniformity of a vectorial Boolean function from $\mathbb{F}_2^n$ to itself are properties which are used to measure the resistance of a function towards linear and differential attacks, respectively. It is well-known that APN and AB functions provide optimal resistance against the said attacks. Many studies are conducted to learn more about the properties of these functions. There are several approaches to the study and characterization of these functions (for example see \cite{Beth,Car2, Car4, Car5}). Today, many APN functions are known as several studies have been conducted on their constructions.  However, not so much is known about APN permutations in even dimensions. It is, therefore, important to study more about the properties and characterizations of APN permutations.  One of the long-standing open problems is to find an APN permutation in dimension $8$. 

%%%%%%%%%%%%%%%%%%%%%%%%%%%%%%%%%%%%%%%%%%%%%%%%%%%%%%%%%%%%%%%%%%%%%%%%%%%%%%%%%%%%%%%%%%%%%%%%%%%%%%%%%%%%%%%%%%%%%%%%%%%%%%%%%%%%%%%%%%%%%%%%%%%

In this paper, we study about characterizations of APN functions (permutations). Can the total weight of the first-order and second-order derivatives of (vectorial) Boolean functions give us any information about the original functions? We first determine a relationship between the weight of a Boolean function and the total weight of its first-order derivatives. The relationship established plays an important role in obtaining most of the results in this paper. We establish characterizations of APN functions (permutations) in terms of the weight of the first-order and second-order derivatives of their components. The total weight of the first-order and second-order derivatives of the functions such as permutations, plateaued, bent, partially-bent, quadratic and balanced functions is determined. 

The paper is organised as follows. In Section \ref{preliminaries}, we report some definitions and known results. In Section~\ref{total-weight-first-derivatives}, we study the weight of the first-order derivatives of Boolean functions and the components of vectorial Boolean functions. We establish characterization of APN permutation in terms of the weight of the first-order derivatives their components. In Section \ref{total-weight-second-derivatives}, we use the total weight of second-order derivatives of Boolean functions and the components of vectorial Boolean functions to characterize bent and APN functions.

\section{Preliminaries}\label{preliminaries}
In this section, we report some definitions and well-known results. The reader is referred to \cite{Ber,Car1, Car2, Cus, Mac, Mus, Nyb1} for more details.

In this paper, the field of two elements, $0$ and $1$, is denoted by $\mathbb{F}_2$.  A vector in the vector space $\mathbb{F}_2^n$ is denoted by small letters such as $v$. We use ordinary addition $+$ instead of XOR $\oplus$. Given any set $A$, we denote its size by $|A|$.

A {\em Boolean function} is any function $f$ from $\mathbb{F}_2^n$ to $\mathbb{F}_2$ and a {\em vectorial Boolean function} is any function $F$ from $\mathbb{F}_2^n$ to $\mathbb{F}_2^m$ for some positive integers $n$ and $m$. In this paper we study vectorial Boolean functions from $\mathbb{F}_2^n$ to $\mathbb{F}_2^n$. A Boolean function in algebraic normal form, which is the $n$-variable polynomial representation over $\mathbb{F}_2$, is given by: \[f(x_1,...,x_n)=\sum_{I\subseteq \mathcal{P}}a_I\left(\prod_{i\in I}x_i\right),\] where $\mathcal{P}=\{1,...,n\}$ and $a_I\in \mathbb{F}_2$. The {\em algebraic degree} (or simply {\em degree}) of $f$, denoted by $\deg(f)$, is $\max_{I\subseteq \mathcal{P}}\{|I|\mid a_I\ne 0\}$.

A Boolean function $f$ is called {\em affine} if $\deg(f)\leq 1$ and {\em quadratic} if $\deg(f)=2$. We denote the set of all affine functions by $A_n$. The {\em weight} of a Boolean function $f$ is defined as $\mathrm{wt}(f)=|\{x\in \mathbb{F}^n\mid f(x)=1\}|$ and $f$ is {\em balanced} if $\mathrm{wt}(f)=2^{n-1}$. The {\em distance} between two Boolean functions $f$ and $g$ is defined as $d(f,g)=\mathrm{w}(f+g)$ and the {\em nonlinearity} of $f$  is defined as $\mathcal{N}(f)=\min_{\alpha \in A_n}d(f,\alpha)$.

A vectorial Boolean function $F$ from $\mathbb{F}_2^n$ to itself is denoted by $F=(f_1,...,f_n)$ where $f_1,...,f_n$ are Boolean functions called {\em coordinate functions}. The functions $\lambda\cdot F$, with $\lambda\in \mathbb{F}_2^n\setminus\{0\}$ and "$\cdot$" denoting dot product, are called {\em components} of $F$ and they are denoted by $F_\lambda$. A vectorial Boolean function $F$ is a {\em permutation} if and only if all its components are balanced. The degree of $F$ is given by $\deg(F)=\max_{\lambda\in \mathbb{F}_2^n\setminus\{0\}}\deg(F_\lambda)$.

The {\em Walsh transform} of a Boolean function $f$ is defined as the function $W_f$ from $\mathbb{F}^n$ to $\mathbb{Z}$ given by: \[W_f(a)=\sum_{x\in\mathbb{F}_2^n}(-1)^{f(x)+a\cdot x}\,,\] for all $a \in \mathbb{F}_2^n$. A set of $W_f(a)$ for all $a\in\mathbb{F}_2^n$ is called {\em Walsh spectrum}.  We define $\mathcal{F}(f)$ as \[\mathcal{F}(f)=W_f(0)=\sum_{x\in \mathbb{F}_2^n}(-1)^{f(x)}=2^n-2\mathrm{wt}(f).\] Notice that $f$ is balanced if and only if $\mathcal{F}(f)=0$.

The nonlinearity of a function $f$ can also be defined as follows: \begin{align}\label{nonlinearity} \mathcal{N}(f)=2^{n-1}-\frac{1}{2}\max\limits_{a\in \mathbb{F}_2^n}|W_f(a)|.\end{align} The nonlinearity of a vectorial Boolean function $F$ is defined as follows: \[\mathcal{N}(F)=\min_{\lambda\in\mathbb{F}_2^n\setminus\{0\}}\mathcal{N}(F_\lambda).\]  For every Boolean function $f$, with $n$ even, \(\mathcal{N}(f)\leq 2^{n-1}-2^{\frac{n}{2}-1}\) and $f$ is said to be {\em bent} if and only if equality holds.  The lowest possible value for $\max_{a\in \mathbb{F}_2^n}|W_f(a)|$ is $2^{\frac{n}{2}}$ and this bound is achieved only for bent functions. 

For $n$ odd, a Boolean function $f$ is called {\em semi-bent} if $\mathcal{N}(f)=2^{n-1}-2^{\frac{n-1}{2}}$. Semi-bent functions can also be defined in even dimension. For $n$ even, a function $f$ is semi-bent if $\mathcal{N}(f)=2^{n-1}-2^{\frac{n}{2}}$. A vectorial Boolean function $F$ in odd dimension is called {\em almost-bent (AB)} if all its components are semi-bent. 

A Boolean function $f$ on $n$ variables is called {\em plateaued} if its Walsh spectrum is either $\{\pm 2^{n/2} \}$ (this happens only when $n$ is even and in this case $f$ is bent) or $\{0,\pm \mu\}$ for some integer $\mu$ (see \cite{Zhe}). In \cite{Cus, Zhe}, the {\em order} of a plateaued function $f$ on $n$ variables is defined as the even integer $r$, $0 \leq r \leq n$, such that all non-zero values of $W_f^2(a)$ are $2^{2n-r}$.

We define the {\em first-order derivative} of $f$ at $a\in\mathbb{F}_2^n$ by $D_af(x)=f(x+a)+f(x)$ and the {\em second-order derivative} of $f$ at $a$ and $b$ is defined as: \[D_bD_af(x)=f(x)+f(x+a)+f(x+b)+f(x+a+b).\] An element $a\in\mathbb{F}_2^n$ is called a {\em linear structure} of $f$ if $D_af$ is constant. The set of all linear structures of $f$ is denoted by $V(f)$ and we call it the {\em linear space} of $f$. It is well-known that $V(f)$ is a vector space.

\begin{theorem}\label{bent-thm}
	A Boolean function $f$ on $n$ variables is bent if and only if $D_af$ is balanced for any nonzero $a\in\mathbb{F}^n$.
\end{theorem}

The two Boolean functions $f,g:\mathbb{F}_2^n\rightarrow \mathbb{F}_2$ are said to be {\em affine equivalent} if there exist an affinity $\varphi:\mathbb{F}_2^n\rightarrow \mathbb{F}_2^n$ such that $f=g\circ \varphi$. This relation is denoted by $\sim_A $ and written as $f\sim_A g$.

\begin{theorem}\label{quadratic}
	Let $f$ be a quadratic Boolean function on $n$ variables. Then
	
	\begin{itemize}
		\item[(i)] $f\sim_A x_1x_2+\cdots +x_{2k-1}x_{2k}+x_{2k+1}$, with $k\leq \lfloor \frac{n-1}{2}\rfloor$, if $f$ is balanced,
		\item[(ii)] $f\sim_A x_1x_2+\cdots +x_{2k-1}x_{2k}+c$,  with $k\leq \lfloor \frac{n}{2}\rfloor$ and $c\in\mathbb{F}_2$, if $f$ is unbalanced.
	\end{itemize}
\end{theorem}

\begin{definition}
	For a vectorial Boolean function $F$ and $a,b\in \mathbb{F}_2^n$, define \[\delta_F(a,b)=|\{x\in \mathbb{F}_2^n\mid D_aF(x)=b\}|.\] The {\em differential uniformity of $F$ } is defined as: \[\delta(F)=\max_{a\ne 0,b\in \mathbb{F}_2^n}\delta_F(a,b)\] and always satisfies $\delta(F)\geq 2$. A function with $\delta(F)=2$ is called {\em Almost Perfect Nonlinear (APN)}.
\end{definition}

There is another representation of vectorial Boolean functions which is known as {\em univariate polynomial representation}. The finite field $\mathbb{F}_{2^n}$ has $2^n$ elements and we write $\mathbb{F}_{2^n}^\ast$ to denote a set of all nonzero elements of $\mathbb{F}_{2^n}$. The vector space $\mathbb{F}_2^n$ can be endowed with the structure of the finite field $\mathbb{F}_{2^n}$ (see \cite{Car1}). Any function $F$ from $\mathbb{F}_{2^n}$ into $\mathbb{F}_{2^n}$ admits a unique univariate polynomial representation over $\mathbb{F}_{2^n}$ given as: 
\begin{align}
F(x)&=\sum_{i=0}^{2^n-1}\delta_ix^i,
\end{align} 
where $\delta_i\in\mathbb{F}_{2^n}$ and the degree of $F$ is at most $2^n-1$. Given the binary expansion $i=\sum_{s=0}^{n-1}i_s2^s$, we define ${\rm w}_2(i)=\sum_{s=0}^{n-1}i_s$. Therefore, we say that $F$ is a vectorial Boolean function whose algebraic degree is given by \(\max\{{\rm w}_2(i)\mid 0\leq i\leq 2^n-1, \delta_i\neq 0\}\) (see \cite{Car1}).  

The {\em (absolute) trace function} $Tr:\mathbb{F}_{2^n}\rightarrow\mathbb{F}_2$ is defined as: 
\begin{align*}
Tr(z)=\sum_{i=0}^{n-1}z^{2^i},
\end{align*} where $z\in\mathbb{F}_{2^n}$. For $\nu\in\mathbb{F}_{2^n}$, a component $F_\nu$ of $F$ is given by $F_\nu(x)=Tr(\nu F)$.

%%%%%%%%%%%%%%%%%%%%%%%%%%%%%%%%%%%%%%%%%%%%%%%%%%%%%%%%%%%%%%%%%%%%%%%%%%%%%%%%%%%%%%%%%%%%%%%%%%%%%%%%%%%%%%%%%%%%%%%%%%%%%%%%%%%%%%%%%%%%%%%%%%%%%%%%%%%%%%%%%%%%%%%%%%%%%%%%%%%%%%%%%%%%%%%%%%%%%%%%%%%%%%%%%%%%%%%%%%%%%%%%%%%%%%%%%%%%%%%%%%%%%%%%%%%%%%%%%%%%%%%%%%%%%%%%%%%%%%%%%%%%%%%%%%%%%%

\begin{theorem}\label{quadratic-weight}
	Let $f$ be any unbalanced quadratic Boolean function on $n$ variables. Then \[{\rm wt}(f)=2^{n-1}\pm 2^{\frac{n+k}{2}-1},\] where $k=V(f)$.
\end{theorem}

\begin{definition}\label{partially-bent-defn}
	A Boolean function $f$ is partially-bent if there exists a linear subspace $E$ of $\mathbb{F}_2^n$ such that the restriction of $f$ to $E$ is affine and the restriction of $f$ to any complementary subspace $E'$ of $E$ (where $E\oplus E'=\mathbb{F}_2^n$) is bent.
\end{definition}

\begin{remark}\label{partially-bent-remark}
	In \cite{Car3}, it is proved that the weight for unbalanced partially-bent function is given by $2^{n-1}\pm 2^{n-1-h}$ where $\dim E=n-2h$ ($E$ as defined in Definition \ref{partially-bent-defn}) and $h\leq n/2$.
\end{remark}

\begin{theorem}\label{APN-fourier-first-derivative}
	Let $F$ be a function from $\mathbb{F}_{2^n}$ into itself. Then \[\sum_{\lambda\neq 0,a\in\mathbb{F}_{2^n}}\mathcal{F}^2(D_aF_\lambda) \geq 2^{2n+1}(2^n-1).\] Moreover, equality holds if and only if $F$ is APN.
\end{theorem}

\begin{theorem}\label{APN-fourier-first_derivative 2}
	Let $F$ be a function from $\mathbb{F}_{2^n}$ into $\mathbb{F}_{2^n}$. Then, for any nonzero $a\in \mathbb{F}_{2^n}$, we have \[\sum_{\lambda\in\mathbb{F}_{2^n}}\mathcal{F}^2(D_aF_\lambda)\geq 2^{2n+1}.\] Moreover, $F$ is APN if and only if equality holds for all nonzero $a\in\mathbb{F}_{2^n}$.
\end{theorem}

\section{On weight of first-order derivatives of APN permutations }\label{total-weight-first-derivatives}
In this section, we consider the total weight of first-order derivatives of Boolean functions and components of vectorial Boolean functions. A characterization of APN permutations in terms of the weight of first-order derivatives of their components is established.

For a given Boolean function $f$, we know from \cite{Car2} that 
\begin{align}\label{square-fourier}
	\mathcal{F}^2(f)=\sum_{a\in\mathbb{F}_2^n}\mathcal{F}(D_af).
\end{align}
We use the relation (\ref{square-fourier}) to show the connection between the weight of a Boolean function $f$ and the total weight of its first-order derivatives.

\begin{proposition}\label{weight-function-weight-derivatives}
	For a Boolean function $f$ on $n$ variables, we have \[{\rm wt}(f)=2^{n-1}\pm \frac{1}{2}\sqrt{2^{2n}-2\sum_{a\in \mathbb{F}_2^n\setminus \{0\}}{\rm wt}(D_af)}.\]
\end{proposition}

\begin{proof}
	From the relation $\mathcal{F}^2(f)=\sum_{a\in\mathbb{F}_2^n}\mathcal{F}(D_af)$, we obtain the following: 
	
	\begin{align}\label{fourier}
	|\mathcal{F}(f)|=\sqrt{\sum_{a\in\mathbb{F}_2^n}\mathcal{F}(D_af)}.
	\end{align}
	
	Since $\mathcal{F}(f)=2^n-2{\rm wt}(f)$, then (\ref{fourier}) becomes 
	\begin{align}\label{weight-derivative}
		|2^n-2{\rm wt}(f)|= \sqrt{\sum_{a\in\mathbb{F}_2^n}\left(2^n-2{\rm wt}(D_af)\right)}.
	\end{align}   
		
	If ${\rm wt}(f)\leq 2^{n-1}$, then $|2^n-2{\rm wt}(f)|=2^n-2{\rm wt}(f)$. So (\ref{weight-derivative}) becomes
	
	\begin{align*}
	 {\rm wt}(f)&= 2^{n-1}- \frac{1}{2}\sqrt{\sum_{a\in\mathbb{F}_2^n}(2^n-2{\rm wt}(D_af))}\\&= 2^{n-1}- \frac{1}{2}\sqrt{2^{2n}-2\sum_{a\in\mathbb{F}_2^n\setminus\{0\}}{\rm wt}(D_af)}.
	\end{align*}
	
	If ${\rm wt}(f)>2^{n-1}$, then $|2^n-2{\rm wt}(f)|=2{\rm w}(f)-2^n$. So (\ref{weight-derivative}) becomes 
						\begin{align*}
							{\rm wt}(f)&=2^{n-1}+ \frac{1}{2}\sqrt{2^{2n}-2\sum_{a\in\mathbb{F}_2^n\setminus\{0\}}{\rm wt}(D_af)}.	\qedhere
						\end{align*}
\end{proof}

By using Proposition \ref{weight-function-weight-derivatives}, we can write the total weight of first-order derivatives of a Boolean function $f$ in terms of weight of $f$ as in the following.
\begin{corollary}\label{sum-weight-derivatives}
	Let $f$ be a Boolean function on $n$ variables. Then we have \[\sum_{a\in\mathbb{F}_2^n\setminus\{0\}}{\rm wt}(D_af)=2{\rm wt}(f)[2^n-{\rm wt}(f)].\]
\end{corollary}

\begin{remark}\label{weight-square-fourier}
	Observe that we can also use the relation \ref{square-fourier} to write the total weight of the first-order derivatives of a Boolean function $f$ in terms of $\mathcal{F}^2(f)$ as follows:
			\[\sum_{a\in\mathbb{F}_2^n\setminus\{0\}}{\rm wt}(D_af)=2^{2n-1}-\frac{1}{2}\mathcal{F}^2(f).\]
\end{remark}

Since $\mathcal{F}(f)=0$ if and only if $f$ is balanced, then by Remark \ref{weight-square-fourier} and by the fact that $0 \leq \mathcal{F}^2(f)\leq 2^{2n}$ we deduce the following.

\begin{corollary}\label{weight-derivative-balanced-function}
	For any Boolean function $f$ on $n$ variables, we have \[\sum_{a\in\mathbb{F}_2^n\setminus\{0\}}{\rm wt}(D_af)\leq 2^{2n-1}.\] Moreover, equality holds if and only if $f$ is balanced.
\end{corollary}

\begin{proposition}\label{sum-weight-derivatives-form}
	For any Boolean function $f$ on $n$ variables, we have \[\sum_{a\in \mathbb{F}_2^n\setminus \{0\}}{\rm wt}(D_af)=2^{2n-1}-2\ell^2\] where $0 \leq \ell \leq 2^{n-1}$. Moreover, $\ell=0$ if and only if $f$ is balanced and $\ell=2^{n-1}$ if and only if $f$ is constant.
\end{proposition}

\begin{proof}
	 Since ${\rm wt}(f)$ is always a positive integer then the quantity in the square root appearing in the relation given in Proposition \ref{weight-function-weight-derivatives} must be \[2^{2n}-2\sum_{a\in \mathbb{F}_2^n\setminus \{0\}}{\rm wt}(D_af)=m^2\] for some positive integer $m$. This implies that \[\sum_{a\in \mathbb{F}_2^n\setminus \{0\}}{\rm wt}(D_af)=2^{2n-1}-2\ell^2\] where $\ell=m/2$. From Corollary \ref{weight-derivative-balanced-function}, it is clear that we must have $0\leq \ell\leq 2^{n-1}$. 
	 
	 Note that, from Corollary \ref{sum-weight-derivatives}, $\sum_{a\in\mathbb{F}_2^n\setminus\{0\}}{\rm wt}(D_af)=0$ if and only if ${\rm wt}(f)=0,2^n$ if and only if $f$ is constant. So it follows that $\ell=2^{n-1}$ if and only if $f$ is a constant. We also conclude, by Corollary \ref{weight-derivative-balanced-function}, that $\ell=0$ if and only if $f$ is balanced.
\end{proof}

\begin{corollary}\label{sum-weight-quadratic-derivatives}
	Let $f$ be any quadratic Boolean function on $n$ variables. Then \[\sum_{a\in\mathbb{F}_2^n\setminus\{0\}}{\rm wt}(D_af)=\begin{cases}
	2^{2n-1} & \text{ if $f$ is balanced}\\
	2^{2n-1}-2^{n+k-1} & \text{ otherwise},
	\end{cases}\] where $k=V(f)$. 
\end{corollary}

\begin{proof}
	If $f$ is balanced, then by Corollary \ref{weight-derivative-balanced-function} we have \(\sum_{a\in\mathbb{F}_2^n\setminus\{0\}}{\rm wt}(D_af)=2^{2n-1}\). Suppose that $f$ is unbalanced.   Then, by Theorem \ref{quadratic-weight}, we know that ${\rm wt}(f)=2^{n-1}\pm 2^{\frac{n+k}{2}-1}$, where $k=V(f)$. Taking ${\rm wt}(f)=2^{n-1}-2^{\frac{n+k}{2}-1}$ and applying Corollary \ref{sum-weight-derivatives} , we have 
	\begin{align*}
			\sum_{a\in\mathbb{F}_2^n\setminus\{0\}}{\rm wt}(D_af)&=2\left( 2^{n-1}-2^{\frac{n+k}{2}-1}\right)(2^n-2^{n-1}+2^{\frac{n+k}{2}-1})\\&=2\left( 2^{n-1}-2^{\frac{n+k}{2}-1}\right)\left(2^{n-1}+2^{\frac{n+k}{2}-1}\right)\\&= 2\left( 2^{2n-2}-2^{n+k-2}\right)=\left( 2^{2n-1}-2^{n+k-1}\right)\\&=2^{n-1}\left( 2^n-2^k\right)=2^{2n-1}-2^{n+k-1}. 
	\end{align*} Using ${\rm wt}(f)=2^{n-1}+2^{\frac{n+k}{2}-1}$ yields the same result.
\end{proof}

Since $D_af$ is balanced for all nonzero $a\in \mathbb{F}_2^n$ if and only if $f$ is bent, then the following result is immediate. 

\begin{corollary}\label{bent-weight-derivatives}
	Let $f$ be a bent Boolean function on $n$ variables, with $n$ even. Then \[\sum_{a\in\mathbb{F}_2^n\setminus\{0\}}{\rm wt}(D_af)=2^{n-1}(2^n-1)=2^{2n-1}-2^{n-1}.\]
\end{corollary}

\begin{theorem}\label{weight-partially-bent}
	For any partially-bent Boolean function $f$ on $n$ variables, we have \[\sum_{a\in\mathbb{F}_2^n\setminus\{0\}}{\rm wt}(D_af)=\begin{cases}
	2^{2n-1} & \text{ if $f$ is balanced}\\
	2^{2n-1}-2^{2n-2h-1} & \text{ otherwise},
	\end{cases}\] where $h\leq n/2$.
\end{theorem}

\begin{proof}
  If $f$ is balanced, then by Corollary \ref{weight-derivative-balanced-function} we have \(\sum_{a\in\mathbb{F}_2^n\setminus\{0\}}{\rm wt}(D_af)=2^{2n-1}\). Now suppose that $f$ is unbalanced. From Remark \ref{partially-bent-remark}, we know that ${\rm wt}(f)=2^{n-1}\pm 2^{n-1-h}$ where $\dim E = n-2h$. Taking ${\rm wt}(f)=2^{n-1}+2^{n-1-h}$ and applying Corollary \ref{sum-weight-derivatives}, we have 
	
	\begin{align*}
		\sum_{a\in\mathbb{F}_2^n\setminus\{0\}}{\rm wt}(D_af)&= 2(2^{n-1}+2^{n-1-h})(2^n-2^{n-1}-2^{n-1-h})\\&=2(2^{n-1}+2^{n-1-h})(2^{n-1}-2^{n-1-h})\\&=2(2^{2n-2}+2^{2n-2-2h})=(2^{2n-1}+2^{2n-1-2h})\\&=2^{n-1}(2^n-2^{n-2h})=2^{2n-1}-2^{2n-2h-1}.
	\end{align*} Using ${\rm wt}(f)=2^{n-1}-2^{n-1-h}$ yields the same result.
\end{proof}

\begin{proposition}
	Let $f$ be a plateaued function of order $r$ on $n$ variables. Then we have \[\sum_{a\in\mathbb{F}_2^n\setminus\{0\}}{\rm wt}(D_af)=\begin{cases} 2^{2n-1} & \text{ if $f$ is balanced}\\
	2^{2n-1}-2^{2n-r-1} & \text{ otherwise}.	\end{cases}\]
\end{proposition}

\begin{proof}
	If $f$ is balanced, we have $\mathcal{F}^2(f)=0$ and if $f$ is not balanced, then $\mathcal{F}^2(f)$ is nonzero positive integer and we deduce from the definition of a plateaued function of order $r$ that it must be $\mathcal{F}^2(f)=2^{2n-r}$. By applying Remark \ref{weight-square-fourier} the result follows.
\end{proof}

\begin{theorem}\label{weight-permutation}
	Let $F$ be a vectorial Boolean function from $\mathbb{F}_{2^n}$ to itself. Then  \[\sum_{\lambda,a\in\mathbb{F}_{2^n}^\ast}{\rm wt}(D_aF_\lambda) \leq 2^{2n-1}(2^n-1).\] Moreover, equality holds if and only if $F$ is a permutation.
\end{theorem}

\begin{proof}
	 By Corollary \ref{weight-derivative-balanced-function}, we can deduce that for any function $F$ from $\mathbb{F}_{2^n}$ to itself, we have \[\sum_{\lambda,a\in\mathbb{F}_{2^n}^\ast}{\rm wt}(D_aF_\lambda)\leq 2^{2n-1}(2^n-1).\] 
	Since all the components $F_\lambda$, $\lambda\in\mathbb{F}_{2^n}^\ast$, of a permutation $F$ are balanced, then it follows that the equality holds if and only if $F$ is a permutation.
\end{proof}

\begin{remark}\label{AB_APN-weight-derivatives}
	Observe that if $F$ is an AB function or APN permutation from $\mathbb{F}_{2^n}$ to itself, then we can use Theorem \ref{weight-permutation} to conclude that the total weight of the first-order derivatives of its components must be equal to \(2^{2n-1}(2^n-1)\). Computer check shows that Dillion's APN permutation in dimension $6$ is equal to $2^{11}(2^6-1)=129024$.
\end{remark}

\begin{remark}\label{weight-APN-1}
	Notice that, from Theorem \ref{weight-permutation}, we can deduce that for any non-bijective APN function $F$ from $\mathbb{F}_{2^n}$ to itself we have \[\sum_{\lambda,a\in\mathbb{F}_{2^n}^\ast}{\rm wt}(D_aF_\lambda) < 2^{2n-1}(2^n-1).\]
\end{remark}

\begin{proposition}\label{weight-first-derivative-APN}
	Let $Q$ be a quadratic APN function from $\mathbb{F}_{2^n}$ to itself, with bent and unbalanced semi-bent components only. Then we have 
	\[\sum_{\lambda,a\in\mathbb{F}_{2^n}^\ast}{\rm wt}(D_aF_\lambda)= 2^{n-1}(2^n-1)(2^n-2).\]
\end{proposition}

\begin{proof}
	 From Theorem \ref{bent-thm}, $D_af$ is balanced for all $a\in\mathbb{F}_{2^n}^\ast$ if and only if $f$ is bent. We know that for any unbalanced quadratic function $f$, $D_af$ is balanced for all $a\in\mathbb{F}^n\setminus V(f)$ and $D_\alpha f=0$ for all $\alpha\in V(f)$. For semi-bent we know that $\dim V(f)=2$, i.e., $|V(f)|=4$. Thus, for each bent, the total weight of its first-order derivatives must be $2^{n-1}(2^n-1)$ (see Corollary \ref{bent-weight-derivatives}) and for each semi-bent, the total weight of its first-order derivatives must be $2^{n-1}(2^n-4)$. For any quadratic APN function with bent and semi-bent components, there are exactly $2/3(2^n-1)$ bent components and $(2^n-1)/3$ semi-bent components (see \cite{Mus}). So the total weight of all first-order derivatives of $Q$ must be 
	 \begin{align*}
	 \sum_{\lambda,a\in\mathbb{F}_{2^n}^\ast}{\rm wt}(D_aQ_\lambda)&=\frac{2}{3}(2^n-1)2^{n-1}(2^n-1)+\frac{1}{3}(2^n-1)2^{n-1}(2^n-4)\\&= 2^{n-1}(2^n-1)\left( \frac{2}{3}(2^n-1)+\frac{1}{3}(2^n-4)\right)\\&= 2^{n-1}(2^n-1)\left(\frac{2(2^n-1)+(2^n-1)-3}{3}\right)\\&= 2^{n-1}(2^n-1)\left(\frac{3(2^n-1)-3}{3}\right)\\&=2^{n-1}(2^n-1)(2^n-2). \qedhere \end{align*}
\end{proof}

 \begin{remark}
 	Computer check shows that the total weight of first-order derivatives of components of the Gold APN power functions in dimensions $4$, $6$ and $8$ satisfy Proposition \ref{weight-first-derivative-APN}. However, the non-quadratic APN power functions such as Kasami functions of dimension $6$ and $8$ also have the total weight that satisfy Proposition \ref{weight-first-derivative-APN}.  
 \end{remark}

\begin{theorem}\label{weight-square-Permutation-derivatives}
	Let $F$ be a permutation from $\mathbb{F}_{2^n}$ to itself. Then \[\sum_{\lambda,a\in\mathbb{F}_{2^n}^\ast}\left[{\rm wt}(D_aF_\lambda)\right]^2\geq 2^{2n-1}(2^n-1)(2^{n-1}+1).\] Moreover, equality holds if and only if $F$ is an APN permutation.
\end{theorem}

\begin{proof}
	We have the following:
	\begin{align}\label{fourier-square-derivatives}
		\sum_{\lambda\neq 0,a\in \mathbb{F}_{2^n}}\mathcal{F}^2(D_aF_\lambda)&=\sum_{\lambda\neq 0,a\in \mathbb{F}_{2^n}}[2^n-2{\rm wt}(D_aF_\lambda)]^2\nonumber\\&=\sum_{\lambda\neq 0,a\in \mathbb{F}_{2^n}}\left[2^{2n}-4\cdot 2^n {\rm wt}(D_aF_\lambda)+4[{\rm wt}(D_aF_\lambda)]^2 \right]\nonumber\\&=2^{3n}(2^n-1)-4\cdot 2^n \sum_{\lambda,a\in \mathbb{F}_{2^n}^\ast}{\rm wt}(D_aF_\lambda)+4\sum_{\lambda,a\in \mathbb{F}_{2^n}^\ast}\left[{\rm wt}(D_aF_\lambda)\right]^2. 
	\end{align} 
	Since $F$ is a permutation then, by Theorem \ref{weight-permutation}, we have \[\sum_{\lambda, a\in \mathbb{F}_{2^n}^\ast}{\rm wt}(D_aF_\lambda)=2^{2n-1}(2^n-1).\]
	So by Equation (\ref{fourier-square-derivatives}) and Theorem \ref{APN-fourier-first-derivative}, we have
	
	\begin{align}\label{weight-square-eqn}
	2^{3n}(2^n-1)-4\cdot 2^{3n-1}(2^n-1)+4\sum_{\lambda,a\in \mathbb{F}_{2^n}^\ast}\left[{\rm wt}(D_aF_\lambda)\right]^2 \geq 2^{2n+1}(2^n-1)
	\end{align}
	from which we deduce that \begin{align}\label{weight-square-eqn-2} \sum_{\lambda,a\in \mathbb{F}_{2^n}^\ast}\left[{\rm wt}(D_aF_\lambda)\right]^2 \geq 2^{2n-1}(2^n-1)(2^{n-1}+1). \end{align} Since, by Theorem \ref{APN-fourier-first-derivative}, we conclude that the equality in the relation (\ref{weight-square-eqn}) holds if and only if $F$ is an APN permutation, then we deduce from the same that the equality in the relation (\ref{weight-square-eqn-2}) holds if and only if $F$ is an APN permutation.\qedhere
\end{proof}

\begin{remark}
	Computer check shows that the result in Theorem \ref{weight-square-Permutation-derivatives} holds for Dillion's APN permutation in dimension $6$. However, the result can never hold for non-bijective APN functions. 
\end{remark}

\section{On weight of second-order derivatives of APN functions}\label{total-weight-second-derivatives}
In this section, we use the total weight of second-order derivatives of (vectorial) Boolean functions to characterize bent and APN functions. 

We first exhibit the relationship between the weight of the first-order derivatives and the weight of the second-order derivatives. By Corollary \ref{sum-weight-derivatives}, we can deduce that for any $a\in \mathbb{F}_2^n$ we have $\sum_{b\in\mathbb{F}_2^n\setminus\{0\}}{\rm wt}(D_bD_af)=2{\rm wt}(D_af)[2^n-{\rm wt}(D_af)]=2^{n+1}{\rm wt}(D_af)-2[{\rm wt}(D_af)]^2$ from which we get 
\begin{align}\label{weight-first-second-order-derivatives}
	\sum_{a,b\in\mathbb{F}_2^n\setminus\{0\}}{\rm wt}(D_bD_af)=2^{n+1}\sum_{a\in\mathbb{F}_2^n\setminus\{0\}}{\rm wt}(D_af)-2\sum_{a\in\mathbb{F}_2^n\setminus\{0\}}[{\rm wt}(D_af)]^2.
\end{align}
From Corollary \ref{weight-derivative-balanced-function}, we know that \(\sum_{a\in\mathbb{F}_2^n\setminus\{0\}}{\rm wt}(D_af)=2^{2n-1}\) if and only if $f$ is balanced. So it implies that the equation (\ref{weight-first-second-order-derivatives}) becomes 
\begin{align}\label{weight-first-second-order-derivatives-balnced}
\sum_{a,b\in\mathbb{F}_2^n\setminus\{0\}}{\rm wt}(D_bD_af)=2^{3n}-2\sum_{a\in\mathbb{F}_2^n\setminus\{0\}}[{\rm wt}(D_af)]^2
\end{align} if and only if $f$ balanced.

By Remark \ref{weight-square-fourier}, for any $a\in\mathbb{F}_2^n$, observe that \(\sum_{b\in\mathbb{F}_2^n\setminus\{0\}}{\rm wt}(D_bD_af)=2^{2n-1}-\frac{1}{2}\mathcal{F}^2(D_af)\) from which we obtain 
\begin{align}\label{weight-second-derivatives-fourier}
	\sum_{a,b\in\mathbb{F}_2^n\setminus\{0\}}{\rm wt}(D_bD_af)=2^{2n-1}(2^n-1)-\frac{1}{2}\sum_{a\in\mathbb{F}_2^n\setminus\{0\}}\mathcal{F}^2(D_af).
\end{align}

Next we give an upper bound on the total weight of the second-order derivatives and show that this bound is met only for bent functions.

\begin{proposition}\label{weight-bent-second-derivatives}
	Let $f$ be a Boolean function on $n$ variables. Then we have \[\sum_{a,b\in\mathbb{F}_{2^n}^\ast}{\rm wt}(D_bD_af)\leq 2^{2n-1}(2^n-1).\] Moreover, equality holds if and only if $f$ is bent .
\end{proposition}

\begin{proof}
	Observe that \[\sum_{a\in\mathbb{F}_2^n\setminus\{0\}}\mathcal{F}^2(D_af)=0\] if and only if $D_af$ is balanced for all $a\in\mathbb{F}_2^n\setminus\{ 0\}$ if and only if $f$ bent. Hence, the result follows immediately from the relation (\ref{weight-second-derivatives-fourier}).
\end{proof}

\begin{proposition}\label{weight-partially-bent-second-derivates}
	For any partially-bent function $f$ on $n$ variables, we have \[\sum_{a,b\in\mathbb{F}_2^n\setminus\{0\}}{\rm wt}(D_bD_af)=2^{2n-1}(2^n-2^{n-2h}),\] where $h\leq n/2$.
\end{proposition}

\begin{proof}
	Observe that, for all $\alpha\in E$, we have $D_\alpha f=c$, with $c\in\mathbb{F}_2$ and $D_af$ is balanced for all $a\in \mathbb{F}^n\setminus E$. Since, by Remark \ref{partially-bent-remark}, we are given that $\dim E =n-2h$, then we have $|\mathbb{F}_2^n\setminus E|=2^n-2^{n-2h}$. Since $D_af$ is balanced then, by Corollary \ref{weight-derivative-balanced-function}, we can deduce that \(\sum_{b\in\mathbb{F}^n\setminus\{0\}}{\rm wt}(D_bD_af)=2^{2n-1}\) and since $D_\alpha f$ is constant then clearly we have \(\sum_{b\in\mathbb{F}^n\setminus\{0\}}{\rm wt}(D_bD_\alpha f)=0\). Therefore, it follows that \begin{align*}\sum_{a,b\in\mathbb{F}^n\setminus\{0\}}{\rm wt}(D_bD_af)&=\sum_{a\in\mathbb{F}^n\setminus E}\sum_{b\in\mathbb{F}^n\setminus\{0\}}{\rm wt}(D_bD_af)+ \sum_{\alpha \in E}\sum_{b\in\mathbb{F}^n\setminus\{0\}}{\rm wt}(D_bD_\alpha f)\\&=2^{2n-1}(2^n-2^{n-2h})+2^{n-2h}(0)=2^{2n-1}(2^n-2^{n-2h}).\qedhere \end{align*}
\end{proof}

It appears like the total weight of second-order derivatives can be associated with the nonlinearity of a Boolean function. Observe from  Propositions \ref{weight-bent-second-derivatives} and \ref{weight-partially-bent-second-derivates} that all bent functions have the same total weight of their second-order derivatives and also all partially-bent functions with the same nonlinearity have the same total weight of their second-order derivatives. 

Next we show that the total weight of the second-order derivatives can be expressed in terms of the Walsh transforms of a Boolean function $f$. 

\begin{lemma}\label{4th power moment - total weight}
	Let $f$ be a Boolean function on $n$ variables. Then we have \[\sum_{a,b\in\mathbb{F}_2^n\setminus\{0\}}{\rm wt}(D_bD_af)=2^{3n-1}-\frac{1}{2^{n+1}}\sum_{a\in\mathbb{F}_2^n} W_f^4(a). \]
\end{lemma} 

\begin{proof}
	If $f$ be a Boolean function on $n$ variables, we have
	\begingroup
	\allowdisplaybreaks
	\begin{align}\label{momentum-fourier-equation}
	\sum_{a\in \mathbb{F}_2^n}W_f^4 (a)\nonumber&=\sum_{a\in\mathbb{F}_2^n}\sum_{x,y,z,w\in \mathbb{F}_2^n}(-1)^{f(x)+f(y)+f(z)+f(w)+a\cdot (x+y+z+w)}\nonumber\\&=\sum_{a\in\mathbb{F}_2^n}\sum_{x,y,z,w\in \mathbb{F}_2^n}(-1)^{f(x)+f(y)+f(z)+f(w)}(-1)^{a\cdot(x+y+z+w)}\nonumber\\&=\sum_{x,y,z,w\in\mathbb{F}_2^n}(-1)^{f(x)+f(y)+f(z)+f(w)}\sum_{a\in \mathbb{F}_2^n}(-1)^{a\cdot (x+y+z+w)}\nonumber\\&=\sum_{x,y,z,w\in\mathbb{F}_2^n |x+y+z+w=0}2^n(-1)^{f(x)+f(y)+f(z)+f(w)}\nonumber\\&=2^n\sum_{x,y,z,w\in\mathbb{F}_2^n |w=x+y+z}(-1)^{f(x)+f(y)+f(z)+f(w)}\nonumber\\&=2^n\sum_{x,y,z\in\mathbb{F}_2^n}(-1)^{f(x)+f(y)+f(z)+f(x+y+z)}\nonumber\\&=2^n\sum_{x,b,c\in \mathbb{F}_2^n}(-1)^{f(x)+f(x+b)+f(x+c)+f(x+b+c)} ~~~~ (\text{ where $y=x+b$ and $z=x+c$}) \nonumber\\&=2^n\sum_{x,b,c\in \mathbb{F}_2^n}(-1)^{D_bD_cf(x)}=2^n\sum_{b,c\in \mathbb{F}_2^n}\mathcal{F}(D_bD_cf)\nonumber\\&=2^n\sum_{b,c\in \mathbb{F}_2^n}\left(2^n-2{\rm wt}(D_bD_cf)\right) \nonumber\\&=2^{4n}-2^{n+1}\sum_{b,c\in \mathbb{F}_2^n}{\rm wt}(D_bD_cf).	
	\end{align}
	\endgroup
	
	The result immediately follows from the relation (\ref{momentum-fourier-equation}).
\end{proof}

\begin{remark}\label{Weight-second-order-derivatives-bent}
	Since $\mathcal{W}_f(a)=\pm 2^{\frac{n}{2}}$, for all $a\in\mathbb{F}^n$, if and only if $f$ is bent, then Proposition \ref{weight-bent-second-derivatives} can also be deduced from Lemma \ref{4th power moment - total weight}.  Since nonlinearity depends on $\max_{a\in \mathbb{F}_2^n}|\mathcal{W}_f(a)|$ (see the relation (\ref{nonlinearity})), then we can use Lemma \ref{4th power moment - total weight} to conclude that a function with high nonlinearity its total weight is relatively high. 
\end{remark}

Next we give a characterization of APN functions that is deduced from Theorem \ref{APN-fourier-first-derivative}.

\begin{theorem}\label{weight-Second-Derivative-APN}
	Let $F$ be a function from $\mathbb{F}_{2^n}$ to itself. Then \[\sum_{\lambda,b,a\in\mathbb{F}_{2^n}^\ast}{\rm wt}(D_bD_aF_\lambda)\leq 2^{2n-1}(2^n-1)(2^n-2).\] Moreover, equality holds if and only if $F$ is APN.
\end{theorem} 

\begin{proof}
By the relation (\ref{weight-second-derivatives-fourier}), for any  $\lambda\in\mathbb{F}_{2^n}^\ast$, we obtain
	\[\sum_{a\in\mathbb{F}_{2^n}^\ast}\mathcal{F}^2(D_aF_\lambda)=2^{2n}(2^n-1)-2\sum_{a,b\in\mathbb{F}_{2^n}^\ast}{\rm wt}(D_bD_aF_\lambda)\] from which we deduce the following:
	\begin{align}\label{equation1}
	\sum_{a\in\mathbb{F}_{2^n}}\mathcal{F}^2(D_aF_\lambda)&=\sum_{a\in\mathbb{F}_{2^n}^\ast}\mathcal{F}^2(D_aF_\lambda)+2^{2n}=2^{2n}(2^n-1)-2\sum_{a,b\in\mathbb{F}_{2^n}^\ast}{\rm wt}(D_bD_aF_\lambda)+2^{2n}\nonumber\\&= 2^{3n}-2\sum_{a,b\in\mathbb{F}_{2^n}^\ast}{\rm wt}(D_bD_aF_\lambda).
	\end{align}
Using the relation (\ref{equation1}), we have 
	
	\begin{align}\label{equation2}
	\sum_{\lambda\neq 0,a\in\mathbb{F}_{2^n}}\mathcal{F}^2(D_aF_\lambda)&= 2^{3n}(2^n-1)-2\sum_{\lambda,a,b\in\mathbb{F}_{2^n}^\ast}{\rm wt}(D_bD_aF_\lambda).
	\end{align}
By applying Theorem \ref{APN-fourier-first-derivative}, the relation (\ref{equation2}) becomes 

\begin{align}\label{equation3}
2^{3n}(2^n-1)-2\sum_{\lambda,a,b\in\mathbb{F}_{2^n}^\ast}{\rm wt}(D_bD_aF_\lambda)\geq 2^{2n+1}(2^n-1),
\end{align} and equality holds if and only if $F$ is APN. The relation (\ref{equation3}) is then reduced to: \[\sum_{\lambda,a,b\in\mathbb{F}_{2^n}^\ast}\text{wt}(D_bD_aF_\lambda)\leq 2^{2n-1}(2^n-1)(2^n-2)\] and equality holds if and only if $F$ is APN.
\end{proof}

\begin{theorem}
	Let $F$ be a function from $\mathbb{F}_{2^n}$ into $\mathbb{F}_{2^n}$. Then, for any nonzero $a\in \mathbb{F}_{2^n}$, we have \[\sum_{\lambda,b\in\mathbb{F}_{2^n}^\ast}{\rm wt}(D_bD_aF_\lambda)\leq 2^{2n-1}(2^n-2).\] Moreover, $F$ is APN if and only if equality holds for all nonzero $a$ in $\mathbb{F}_{2^n}$.
\end{theorem}

\begin{proof}
	By applying the relation (\ref{square-fourier}), we have the following:

	\begin{align}\label{weight-sd-2}
		\sum_{\lambda\in\mathbb{F}_{2^n}}\mathcal{F}^2(D_aF_\lambda)&=\sum_{\lambda,b\in\mathbb{F}_{2^n}}\mathcal{F}(D_bD_aF_\lambda)\nonumber\\&=\sum_{\lambda,b\in\mathbb{F}_{2^n}}[2^n-2{\rm wt}(D_bD_aF_\lambda)]\nonumber\\& = 2^{3n}-2\sum_{\lambda,b\in\mathbb{F}_{2^n}}{\rm wt}(D_bD_aF_\lambda)\nonumber\\&=2^{3n}-2\sum_{\lambda,b\in\mathbb{F}_{2^n}^\ast}{\rm wt}(D_bD_aF_\lambda).
	\end{align} By Theorem \ref{APN-fourier-first_derivative 2} and the equation (\ref{weight-sd-2}) we obtain the following:
	\[\sum_{\lambda,b\in\mathbb{F}_{2^n}^\ast}{\rm wt}(D_bD_aF_\lambda)\leq 2^{2n-1}(2^n-2). \qedhere\]
\end{proof}

\begin{remark}
	As observed in Remark \ref{Weight-second-order-derivatives-bent}, we can conclude that the nonlinearities of components of APN functions are relatively high since in Theorem \ref{weight-Second-Derivative-APN} the total weight of its second-order derivatives is the highest.
\end{remark}

\section{Conclusion}

In this paper, we have established a characterization of APN permutations in terms of the  weight of their first-order derivatives. Using the total weight of second-order derivatives, a characterization for APN functions has been discovered. The total weight of the first-order and second-order derivatives for the functions such as bent, partially-bent, plateaued, balanced and permutations have been determined.

\end{document}